\documentclass{IEEEtran}

\usepackage{color,amsmath, amssymb,mathrsfs}
\usepackage{verbatim}
\usepackage{graphicx}
\usepackage{cite}
\usepackage{pdfsync}
\usepackage{epstopdf}
\makeatletter
\def\ps@headings{%
\def\@oddhead{\mbox{}\scriptsize\rightmark \hfil \thepage}%
\def\@evenhead{\scriptsize\thepage \hfil \leftmark\mbox{}}%
\def\@oddfoot{}%
\def\@evenfoot{}}
\makeatother
\newcommand{\ignore}[1]{}

\newtheorem{definition}{Definition}
\newtheorem{theorem}{Theorem}

\newtheorem{lemma}{Lemma}

\newtheorem{proposition}{Proposition}

\newcommand{\mf}{\mathbf}

\begin{document}
\def\nat{\mathbb{N}}
\def\ff{\mathbb{F}}

\title{On the Security of Key Extraction from Measuring Physical Quantities}

\author{\IEEEauthorblockN{Matt Edman\IEEEauthorrefmark{4},
Aggelos Kiayias\IEEEauthorrefmark{1},\IEEEauthorrefmark{2},
Qiang Tang\IEEEauthorrefmark{3}
%\IEEEauthorrefmark{1},\IEEEauthorrefmark{2},
and B{\"u}lent Yener\IEEEauthorrefmark{4}}

\IEEEauthorblockA{\IEEEauthorrefmark{1}National and Kapodistrian University of Athens, Greece}\\
\IEEEauthorblockA{\IEEEauthorrefmark{2}University of Connecticut, U.S.A}\\
\IEEEauthorblockA{\IEEEauthorrefmark{3}Cornell University, U.S.A}\\
\IEEEauthorblockA{\IEEEauthorrefmark{4}Rensselaer Polytechnic Institute, U.S.A}
% <-this % stops an unwanted space
%\thanks{hello}}
\thanks{This work is supported by NSF grant 0831304, and it was mostly done while the third author was at the University of Connecticut.} }

\maketitle

\ignore{%%%%%%%%%%
\author{Matt Edman\inst{1}, Aggelos Kiayias\inst{2}, Qiang Tang\inst{2} and B{\"u}lent Yener\inst{1} }
\institute{Rensselaer Polytechnic Institute
\\\{edmanm2,yener\}@cs.rpi.edu\and University of Connecticut \& University of
Athens\\\{aggelos,qiang\}@cse.uconn.edu} 

While ad hoc lower bound
assumptions on the
conditional entropy  of Alice and
Bob's measurements given the view of the
adversary  are sufficient to argue the security of the key, it is very hard to
{\em experimentally} verify that
 such assumptions are indeed true.
}%%%%

\begin{abstract}
Key extraction via measuring a physical quantity
is a  class of information theoretic
key exchange protocols that rely on the physical characteristics of
the communication channel, to enable the computation of a shared key
by two parties that share no prior secret information. The
key is supposed to be information theoretically hidden to an
eavesdropper. Despite the recent surge of research activity in the
area, concrete claims about the security of the protocols typically
rely on channel abstractions that are not fully experimentally
substantiated. In this work, we  propose a novel methodology
for the {\em experimental} security analysis of these
protocols. The crux of our methodology is  a falsifiable
channel abstraction that is accompanied by an efficient
experimental approximation algorithm of the {\em conditional min-entropy}
available to the parties given the view of the eavesdropper.

We  focus on the signal strength between two wirelessly
communicating transceivers as the measured quantity and we use an
experimental setup to compute the conditional min-entropy of the
channel given the view of the attacker which we find to be linearly
increasing. Armed with this understanding of the channel, we
showcase the methodology by providing  a general protocol  for key
extraction in this setting that  is shown to be secure for a
concrete parameter selection. In this way we provide a 
comprehensively  analyzed wireless key extraction protocol that is
demonstrably secure against passive adversaries assuming our falsifiable
channel abstraction. Our  
use of hidden Markov models as the channel model and a dynamic
programming approach to approximate conditional min-entropy might be of independent interest, while
other possible instantiations of our methodology can be feasible and may be motivated by this work.
%Our approach entails  a number of features that
%include a new quantization algorithm for transforming signal
%measurements to bitstrings that utilizes low distortion embeddings
%in metric spaces, as well as an experimental approach for estimating
%the conditional min entropy  that remains in the channel in the
%presence of an eavesdropper which uses the Viterbi algorithm and the
%forward algorithm for hidden Markov models. Among others, we also demonstrate
%experimentally for the first time the linear growth of conditional min-entropy
%of  physical layer measurements given the view of the passive eavesdropper.
\end{abstract}

\section{Introduction}
\label{sec:introduction}

Key extraction between two parties, Alice and Bob, by measuring
a physical quantity is based on two basic premises. First, due to
the   physical properties  of
the  quantity measured, Alice and Bob obtain highly
correlated measurements. Second, the   measurements obtained  by the
eavesdropper are weakly correlated. This hypothesized gap of
the correlation of measurements between Alice and Bob and the
correlation of the eavesdropper's measurements opens the door for a
key exchange mechanism that would have to perform an information
reconciliation~\cite{BS94} and privacy amplification~\cite{BBR88}
step to enable the calculation of a secret key that is almost
independent of the adversary's view.

%Nevertheless, there is an important challenge that distinguishes
%physical-layer key extraction from theoretical information
%theoretic key exchange or  quantum key exchange protocols.
%%
In the literature, one can distinguish two classes of such
protocols. The first class is practice-oriented protocols in which
the security of the key follows directly from certain strong assumptions or is based (at best)  on statistical tests.
The second class  is   theoretical protocols  in which broad channel
abstractions are made and then these are formally shown to imply a
correlation gap between the adversary and the two parties, thus
showing information theoretic security. Nevertheless, despite that
security is proven formally, for a real-world implementation, a leap
of faith is still required to accept the fact that the channel
abstraction is indeed capturing the real-world setting of the
adversary. The only exception is  quantum key-exchange, where there
are results from quantum mechanics such as the no-cloning theorem
that do provide the underlying formal basis justifying the
connection between the channel abstraction and the real-world.
%(even though
%there are also doubts for this as well \cite{anderson-brady}).
The
lack of   strong general arguments in
key extraction from measuring physical quantities
means that any deviation
of the real-world adversary setting  from the assumed channel
abstraction will lead to a security breakdown and thus to a dubious
security status for the key extraction protocol.

In this work we take a first step towards addressing  this fundamental problem 
%of bridging the
%gap between experimental analysis and theoretical security arguments,
by providing a novel formal security methodology that enables us
to argue experimentally about the security of such protocols. In
a nutshell, we put forth a general framework that  (i)
enable us to mathematically estimate the conditional entropy in the key extracted from physical measurements; (ii)
the security arguments can be tied to a specific real-world
setting through falsifiable (but not oversimplified) assumptions.
% that are
%verifiable through experiments.

\ignore{%%%%%%%%%%%%%%%passive-active model
The authors of~\cite{ESWM12} also considered how to deploy a
man-in-the-middle attack to recover part of the key generated from
the wireless signals, and~\cite{ZAS12,Clark12} discussed the problem
of extracting a key from wireless signal in an adaptive model where
the adversary is active, she not only eavesdrops the communication
but also tries to inject signals or to do anything she is able to
influence the communication. We would like to remark that although
we consider security in the passive model, our main objective is to
derive a framework of estimating the conditional entropy contained
in the physical quantities received by the communicating parties,
this problem will still be fundamental when addressing the security
in the adaptive model, and one may use our framework as a
steppingstone to study this problem in the adaptive model;
furthermore, we think it still meaningful to resolve this central
issue in the passive model, as an example of bridging the gap
between theoretical argument and experimental evaluation w.r.t. the
security of extracting keys from physical quantities.
}%%%%%%%%%%%%%%%%
%This enables to provide concrete security arguments for key extraction
%that are experimentally verifiable.
%As part of our results we also provide a new general protocol for key
%extraction from physical measurements which is the basis for our security
%arguments. Finally, we apply our methodology and we
%obtain a  secure key extraction protocol that relies on
%measurements of physical-layer properties in wireless communication.

%We note that there are substantial benefits to be gained from key
%extraction protocols from measuring physical quantities, if their
%security can be properly quantified in a given setting.
%%Indeed, such
%%protocols, contrary to quantum key-exchange or other
%%information-theoretic  key exchange systems that rely on other
%%assumptions, can be implemented very efficiently as a  functionality
%%of the physical communication layer; in this way,
%Two communicating transceivers can extract a key at effectively no
%amortized cost by simply engaging in regular operation, and
%measuring the physical quantity as a byproduct. (e.g, in the
%wireless setting, communicating in plaintext and measuring the
%signal strength of the messages exchanged (independently of
%content)).

%Also, our security argument is also applicable in other key
%extraction protocol from measuring a sequence of physical
%quantities.

% !TEX root =  wir03.tex

\subsection{Related Work and Our Results}
\label{sec:related-work}

Information-theoretic treatment of secure key exchange was initiated
by Wyner~\cite{Wyner} in the wiretap model and it was later studied by Maurer \cite{Maurer93} and Ahlswede, Csiszar \cite{AC93} in the common random source model. These works dealt with the feasibility of
key exchange assuming a non-zero ``equivocation-rate'', a measure
expressing the  uncertainty of the eavesdropper about the parties'
communication.  Beyond generic  schemes, for specific physical
quantities, there has already been  a series of theoretical studies,
in quantum key exchange \cite{BB84}, and wireless key exchange
focusing on the ``secrecy capacity'' between wirelessly
communicating
transceivers, (e.g.,~\cite{BBDM08}, following the wiretap channel and \cite{BR06,KDW11,GLE08,LLP09,Clark12, HHY95, HSHC96}, following the shared randomness channel.) These
works, being conceptual only \cite{HHY95, HSHC96}, or information theoretic in nature, contain no
experimental justification of the channel model they utilize: for
the quantum ones, the underlying physics justify the model but the
wireless ones have no real-world justification, most of them relying on
the assumption that the signal measurements are independent across
time and thus any existing non-zero equivocation rate can be  magnified for a suitable number of
transmissions straightforwardly.

Considering works that were more experimental in nature, there were
several practical algorithms suggested for wireless key extraction
~\cite{TM01,JPKPCK09,SKMY07,PCJSK10,YMRSTM10,LR05,WTS07,AHOKS05,MTMYR08,WSKR11,WZM11}.
A common characteristic of these works is that significant attention
is given to demonstrate correctness (i.e., that Alice and Bob
calculate the same key), and efficiency; much less formal analysis (as we detail below)
is performed to demonstrate security. 
% that the eavesdropper
%cannot extract {\em any} information about the key). 
More specifically, the security of many of these works, e.g., \cite{MTMYR08, JPKPCK09, PCJSK10, YMRSTM10}, directly relied on the assumption that if the adversary is a wavelength away from the communicating parties, the adversary's observation is independent of the parties' measurements. 
%This assumption is strong and 
%also can not be justified by the experiments.
%actually some real world attacks were demonstrated (see below \cite{EKY11,DLMA10}).
%in~\cite{TM01,WTS07}, no proper security analysis is provided whatsoever, 
Some other works made some further steps to analyze security. For instance, \cite{MTMYR08,WSKR11} consider performing
randomness tests but these are insufficient as security guarantees
(such ``random'' sequences are not necessarily unpredictable).
In~\cite{LR05} a type of attacks based on blind deconvolution is
considered and it is argued experimentally that such attacks are
unlikely to apply; however, this does not preclude other types of
attacks. Works such as~\cite{AHOKS05,SKMY07} provided informal
security arguments based on calculating the correlation between the
eavesdropper's measurements and the parties' measurements. However, such correlation calculation can not guarantee  the amount of (conditional) entropy contained in the extracted key. 
%Unfortunately, correlation measures used such as the Pearson
%correlation coefficient are insufficient for arguing about security.
%Indeed, even if no linear correlation exists, this does not preclude
%the existence of more complex non-linear correlation functions
%that can be taken advantage of in an attack and the existence of
%such functions cannot be precluded by these arguments.

In fact, the lack of properly rigorous security claims can lead to
(partial) key recovery attacks as demonstrated
in~\cite{EKY11,DLMA10} that examined known protocols in a specific
deployment. These works exemplified the fact that even though
existing analysis has demonstrated successfully that the signal
observed by the two parties has sufficient entropy it remains open
whether the {\em conditional entropy} on the eavesdropper's view is
non-zero in a specific real-world deployment.

A number of techniques for reconciling the errors between the two
communicating parties have been shown in the literature. For real valued
measurements (as in the wireless key exchange setting),
a popular
approach employs a ``thresholding''
methodology~\cite{SKMY07,MTMYR08} where the signal is transformed to
a bitstring using only the level of measurements where, for example,
deep fades are observed. In such settings it can be shown that the
errors are so few that reconciliation becomes easy. Unfortunately,
such reconciliation requires interaction that may result in
potentially   zeroing the conditional entropy, a fact that
experimentally remains open. Other approaches suffer from similar
problems, such as~\cite{JPKPCK09} that relies on the Cascade
protocol or~\cite{AHOKS05} that uses specialized antennas.

We conclude that information theoretic security in all
these previous works relies on whether the channel
abstraction fits the real-world adversary setting
and whether any
additional reconciliation
interaction performed does not cut from the conditional entropy
too much.
While ad hoc lower bound
assumptions on the
conditional entropy  of Alice and
Bob's measurements given the view of the
adversary  are sufficient to argue the security of the key, it is very hard to
{\em experimentally} verify that
 such assumptions are indeed true. Indeed, in order to apply the textbook
calculation of conditional entropy, one would have to carry
experiments an immense number of times. Even worse, one  can
formally prove  that an exponential number of samples are necessary
for approximating conditional entropy in the general case (this can
be derived from \cite{valiant11}).
%\footnote{It is easy to construct two
%pairs distributions $(X,Y), (X',Y')$ so that $H(X|Y)=0$ and $H(X'|Y')$ is maximal, while
%the pairs are indistinguishable given any polynomial number of samples.}.

In this work we address the above problems with the following
contributions.
\begin{itemize}
\item First, we introduce a methodology to experimentally argue about the security
of  key exchange protocols that are based on measuring real-valued
physical quantities in the passive model: The crux of our
methodology is a falsifiable channel abstraction that is accompanied
by an efficient experimental approximation algorithm of the
conditional min-entropy of the two parties given the view of the
eavesdropper.

Specifically, we  model  measurements and adversarial observations
via a hidden Markov model (HMM) and we   utilize a dynamic
programming approach to estimate the conditional min-entropy \footnote{There are works studying the capacity of  Markovian channels \cite{GV96}, or using HMM to model the package loss process to infer channel parameters \cite{SV01}.  To the best of our knowledge, our work, for the first time, uses HMM to model channels subject to adversarial eavesdropping,   for the purpose of deriving  conditional entropy lower bounds for security guarantees.} : we
achieve that through a combination of the Viterbi
algorithm~\cite{viterbi67} and the forward algorithm~\cite{BE67}.
This algorithmic approach is beneficial as it allows with only a
polynomial number of experiments  in the size of the key to argue
about the conditional min-entropy {\em without} oversimplifying the
channel abstraction (e.g., assuming measurements are mutually independent across time~\cite{BR06,KDW11,GLE08,BBDM08,LLP09,Clark12}).\footnote{Note that we are note claiming HMM exactly models how channel behaves, instead, allowing correlation across signals is an important first step towards the real world model departing from those idealized assumptions.} Furthermore, the Markovian nature of the
channel is falsifiable and is possible to verify   it experimentally. Given the conditional min-entropy bound,  we eventually show
security through the calculation of the min-entropy loss during the
standard protocol steps of quantization,  privacy amplification and
information reconciliation. We showcase  the result by providing a
general protocol following these three steps that relies on
``low-distortion'' embeddings~\cite{GKL03}, secure sketches,
\cite{DRS04} and randomness extractors~\cite{NZ96}. We note that
the HMM abstraction is only one out of many  ways to instantiate our
methodology for arguing about security and there could be others
that may be discovered motivated by our work. 
%Also, we do not claim that our model of the Markovian channel holds unconditionally, rather it is strictly weaker then previous assumptions, and approximate the actual channel better xxxxxx).

\item Second, we apply the methodology we put forth above
to the setting of wireless key exchange, where
the signal strength is the target physical quantity that the two parties measure.
We designed an experimental configuration that enabled us to measure the conditional
min entropy as well as the correctness of channel abstraction our methodology utilizes.
Our experimental configuration included a robotic mobile device equipped with various
transceivers that performed  measurements of signal strength over long periods of time.
Using the data and our methodology we calculated the  conditional min-entropy  and error
rate  over time, showing the feasibility of secure key exchange for the type of adversaries
used in our experimental setup. Our experimental results enabled us to identify the
concrete parameters needed for generating a secure key of a specific length and also
predict how these parameters would need to change for generating longer keys.
\end{itemize}
%
%{\bf Organization. } In section~\ref{sec:key-generation}, we present
%a conditional entropy learner for Markovian processes,  and a general
%protocol for key extraction from measuring physical quantities. Then
%in section~\ref{sec:wireless-protocol}, we apply the methodology to
%the setting of wireless key exchange,  experimentally determine
%the relevant parameters and instantiate the key extraction protocol.
%Finally, we experimentally justify all the assumptions needed about the channel   in the appendix.

We note that our work focuses on demonstrating feasibility of analyzing security  for physical layer key extraction from experimentally falsifiable assumptions; while our  instantiation is practical, we leave for future work the further optimization of  efficiency within  our framework. 
Finally we stress that a more ambitious objective would be to provide a framework for security
in the active adversarial model; indeed, 
in  key extraction protocols  certain types of active attacks have been demonstrated, e.g., \cite{ESWM12,ZAS12,Clark12}. While this is beyond the scope of the present work, 
a framework such as ours that provides a way to 
lower bound the conditional entropy available to the two transceivers 
can be a fundamental intermediate step towards a formal treatment  of security
in the active model. 

%Resolving this central issue in the passive model is an important first step, in that one may use our framework as a
%steppingstone to experimentally  argue  security in the adaptive model. %furthermore, we think it still meaningful to resolve this central
%issue in the passive model, 
%as a first step towards experimentally arguing the security of key extraction protocols from measuring physical quantities .
%bridging the gap
%between theoretical argument and experimental evaluation w.r.t. the
%security of extracting keys from physical quantities.

\section{Background of Key Extraction from Measuring Physical Quantities}
\label{sec:key-generation}

\subsection{Definitions}

The problem of key extraction from physical measurements can be
abstracted as a game between two players, Alice and Bob, who have
access to their own separate measurement devices for a certain
physical quantity. The physical quantity itself is generated through
the actions of Alice, Bob and potentially also the adversary, Eve
and for the purpose of this section will remain purposefully
undetermined (we will detail a specific instantiation in
section~\ref{sec:wireless-protocol}).

We will use  $X_1,\ldots,X_n$ to denote the measurements obtained by
Alice, and $X'_1,\ldots, X'_n$ the measurements obtained by Bob.
Finally the adversary is also able to obtain measurements of the
same physical quantity denoted by $Y_1,\ldots,Y_n$. The value $n$ is
arbitrary and the algorithm for extracting the key should allow any
choice for this value. The mechanism for determining the appropriate
value of $n$, given a certain level of security that needs to be
attained is, in fact, an important part of the   protocol design
problem. Beyond access to the measurements, we assume also that the
parties, Alice and Bob have the ability to engage in ``public
discussion.'', i.e., they can utilize an authenticated channel. The
implementation of this channel is separate and orthogonal to our
objectives.

Given the above,   a
$(l,\epsilon_{c},\epsilon_{u})$ key generation system is a protocol between
Alice and Bob running
over an authenticated channel so that each party has access to
their physical measurement devices and satisfies
 these properties:
\begin{itemize}
\item (Correctness) Alice and Bob both output the same $l$-bit
key with probability $1-\epsilon_{c}$.
\item (Security) Conditional to the view of any adversary
the key calculated by the two parties has statistical distance
from the uniform distribution over
$\{0,1\}^l$ at most $\epsilon_{u}$.
\end{itemize}

%%%%%%%%%%%%%%%%%%%%%%%%%%%%%%%%%%%%%%%%%%%%%%%%%%%%%%%%%%%%%%%%%

%%!TEX root = wir03.tex

%\subsection{Learning Conditional Min-Entropy}

The most important consideration in the process of describing how to
extract the key from the physical measurements performed by Alice
and Bob is determining the amount of uncertainty that exists in the
measurements conditioned on the view of the adversary -- if no
sufficient uncertainty exists then they cannot be expected to
complete the key extraction securely. We first recall the standard
definitions of metrics for uncertainty: min-entropy and conditional
min-entropy.
%, then  that starts from a experimentally
%falsifiable physical assumption, and argue the security through a
%logically sound mathematical reasoning.
%Before we propose the key extraction protocol, we first examine how
%to estimate how much uncertainty are there in the physical
%measurements received by the players when there is also an adversary
%eavesdropping partial information about the measurements. The
%uncertainly will form a basis for any later security argument.
%Instead of relying on the theoretical assumptions, such as that the
%eavesdropped measurements are independent of the measurements which
%are unjustifiable in realistic settings, we would like to move one
%step forward to bridge the gap between theoretical security argument
%and experimental security. We introduce the problem of learning the
%conditional min-entropy given access to the measurement devices of
%one of the players and the adversary. We define min-entropy and
%conditional min-entropy first.

\begin{definition}
\label{def:ave-entropy}
  The min-entropy $\textbf{H}_{\infty}(A)$ of a random variable A is defined
  as $\textbf{H}_{\infty}(A) = -\log(\max_{a}\Pr[A=a])$. The conditional min-entropy
  of $A$ on the event that another random variable $B$ equals a specific value $b$ is
  $\textbf{H}_{\infty}(A|B=b)=-\log(max_{a}\Pr[A=a|B=b])$.
  Further, we define the (average-case) conditional  min-entropy of $A$ given $B$ as
  $\widetilde{\textbf{H}}_{\infty}(A|B) =
    -\log(E_{b\leftarrow B}[2^{-\textbf{H}_{\infty}(A|B=b)}])$ \footnote{See \cite{DRS04} for detailed justifications of the definition of average-case conditional min-entropy.}.
 \end{definition}

%%%%%%%%

% !TEX root = wir03.tex

\subsection{A General Protocol}
\label{sec:generalprotocol}
Next, we will present a general protocol for our key generation
process from measuring physical quantities. It has three basic
steps. First, $A$ and $B$ produce a  sequence of events and their
corresponding measurements and then convert them to bitstrings
denoted by $\rho_{A}$
and $\rho_{B}$, respectively. Subsequently they perform
 information reconciliation and  privacy
amplification. All these techniques are fairly standard,  we choose them properly so that our protocol  enable  us to bound the entropy loss
that takes place  during each step of  the execution of the algorithm and
ensures the proper extraction of the key.

\noindent{\bf Bit Quantization}
\label{subsec:quantization}
%is accomplished by
%a simple {\tt ping}-like protocol that is described in more detail in
%Section~\ref{subsec:experiment-setup}.
Each party has at
its disposal a series of measurements $X_1,\ldots, X_n$.
% where $X_i$ corresponds to the average signal strength that
%was observed for the $i$-th transmission of the {\tt ping} protocol.
%There have been various approaches proposed in previous literature
%to convert (or {\it quantize}) a sequence of physical quantities
%$X_1, X_2, \ldots, X_n$  to a bitstring $b_1, b_2, \ldots, b_N$.
%Without loss of generality we assume that $X_i$ are signed integers.
%Most approaches (like the ones using wireless channel
%measurements)are essentially thresholding algorithms, like the ones
%proposed by Azimi-Sadjadi et al.~\cite{SKMY07}, Mathur et
%al.~\cite{MTMYR08}, Jana et al.~\cite{JPKPCK09} and others. Adaptive
%thresholding aims to reduce the number of bit differences (i.e., the
%Hamming distance) between the bitstrings that are output by two
%communicating parties.
We put forth a quantization approach that
utilizes a low distortion embedding $\tau$ from $\ell_1$-distance
metric space into Hamming distance metric space. The
$\ell_1$-distance (or Manhattan distance) between two vectors
$\textbf{X},\textbf{Y}$ in an $n$ dimensional real vector space is
defined as the sum of the absolute difference between corresponding
coordinates, i.e, $\ell_1(\textbf{X},\textbf{Y})=\sum_{n}|X_i-Y_i|,$
where $\textbf{X}=(X_i,\ldots,X_n),\textbf{Y}=(Y_1,\ldots,Y_n)$.

An embedding is a mapping between two metric spaces. It is said to
have distortion $c$ if and only if the distance is preserved with up
to a multiplicative factor of $c$~\cite{OR05}. Given a low distortion embedding
$\tau$, we not only quantize the physical measurements, but also
prepare two bitstrings to have a small Hamming distance which is
critical for the reconciliation algorithm. Last, it is not hard to
construct a low distortion embedding, one can verify the following
proposition that a simple unary encoding is a good low distortion
embedding.
\begin{proposition}
Suppose $B$ is a finite set of integers, 
%for any $i$, $X_i$ is a random variable over a finite set
%$B\subseteq\mathbb{Z}$ 
and $m =\max \{ |x| \mid x\in B\}$, 
%where $|x|$ 
%and $m$ is the maximum bit length of elements in $B$. 
Consider the
mapping $\tau: B\rightarrow \{0,1\}^{m}$  defined by applying a unary
encoding, i.e, $\tau(x)$ is a bitstring of length $m$, and it is composed of  
%bit represents $sign(x)$, followed by 
$m-|x|$ 
consecutive $0$s followed by $|x|$ consecutive $1$'s.  Then, $\tau$ is an embedding
with distortion at most 1.
\end{proposition}

\noindent{\bf Reconciliation using Secure Sketches}
\label{subsec:reconciliation} After collecting measurements
and performing the bit quantization step, Alice and Bob possess two
$N$-bit strings $\rho_A,\rho_B$, respectively. For now, we assume
that there exists a $d \in \nat$ such that with overwhelming
probability the Hamming distance satisfies $H(\rho_{A}, \rho_{B})
\leq d$.
%Note
%that if $e = \rho_A\oplus \rho_B$, it holds that $W(e) \leq d$ where
%$W(\cdot)$ is the Hamming weight of a string (i.e., the number of nonzero bits).
%
The value of $d$ in specific protocol can be determined
experimentally for a given physical quantity, (we do that in the next section,
see Figure~\ref{figure:linear-error} and
Table~\ref{table:estimation}). We  next describe our
information reconciliation protocol that can be proven correct under
the above definitions.

First, recall that an $\langle n,m,m',d\rangle$-secure sketch ~\cite{DRS04} is a
pair of randomized procedures, ``sketch'' (\texttt{SS}), and
``recover" (\texttt{REC}), such that
\texttt{SS} on input $n$-bit string $w$, returns a bit string $s\in
\{0,1\}^{*}$, and \texttt{REC} takes an
$n$-bit string $w'$ and a bit string $s\in \{0,1\}^{*}$, and outputs
an $n$-bit string.

\texttt{SS} and \texttt{REC} satisfy \emph{correctness} and
\emph{security}. In particular, correctness states that  if $H(w,w')\leq d$, then
$\texttt{REC}(w',\texttt{SS}(w))=w$; while security means
that for any distribution $W$ with
min-entropy $m$ over $\{0,1\}^{n}$, the value of $W$ can be
recovered by an adversary who observes $s$ with probability no
greater than $2^{-m'},$ i.e,
$\widetilde{\textbf{H}}_{\infty}(W|SS(W))\geq m'$.

Achieving information reconciliation between two parties holding
data $w,w'$ respectively using sketches is simple. The first party
transmits $s =\texttt{SS}(w)$ and the second computes
$\texttt{REC}(w',s)$. Various secure sketches can be designed
depending on the metric space. We next recall two secure sketches
over the Hamming metric. Let $\mathcal{C}$ be an error correcting
code over $\ff_{2}^{n}$ which corrects up to $d$ errors and
$|\mathcal{C}|=2^{k}$. Thus, there is an algorithm $\mathsf{Dec}$
that given any $v\in \mathcal{C}$ and any $e\in \{0,1\}^n$ with
$Wt(e) \leq d$, it holds that $\mathsf{Dec}(v\oplus e) = v$, where
$Wt(\cdot)$ stands for the Hamming weight of $e$. Observe that under
our assumption $H(\rho_A,\rho_B)$ is at most $d$ and thus the output
of Bob is exactly $\rho_A$. It follows that Alice and Bob recover
$\rho_A$ and thus information reconciliation is achieved.

We note that the communication overhead of the above protocol is
unnecessarily high. An improvement can be achieved by using a more
communication efficient secure sketch. For a $(n,k)$ linear code,
the $n\times(n-k)$ parity check matrix $\mathbf{H}$ has the property
that for every codeword $c$, $c\cdot\mathbf{H}=0$. For any
transmitted message $y=x+e$ of length $n$, the syndrome is defined
as $\texttt{syn}(y)=y\cdot\mathbf{H}$. Assuming that the
error-codeword $e$ can be easily computed using $\texttt{syn}(e)$
(doing what is known as syndrome-decoding) the length of the sketch
is only $n-k$.
%thus achieving a factor of $1-\frac{k}{n}$ improvement
%in the transmission.
Bob computes the syndrome for his own bits, $\texttt{syn}(\rho_B)$,
and finally computes
$\texttt{syn}(e)=\texttt{syn}(\rho_B)-\texttt{syn}(\rho_A)$, from
which he obtains $e$, and thus $\rho_A=\rho_{B}-e$. It is not hard
to see that the syndrome based construction is an
$(n,m,m-n+k,d)$-secure sketch, from the chain rule described in
definition~\ref{def:ave-entropy}.

\noindent{\bf Privacy Amplification} After both parties have derived
the same bitstrings, they must ``purify'' the strings to make the
joint output suitably random as a cryptographic key. For this task,
we employ   a randomness extractor~\cite{NZ96}. We first recall the
definition of statistical distance which characterizes how similar
two distributions are. Then the level of randomness of a random
variable can be measured by the statistical distance of the variable
from the uniform distribution. The statistical distance of two
probability distributions $X_{1},X_{2}$ with support $S$ is:
$\max\limits_{T\subseteq S}\{Pr[X_{1}\in T]-Pr[X_{2}\in T]\}$, denoted by
$SD(X_{1},X_{2})$. It can be also defined as $\frac{1}{2}\sum_{s\in
S}|Pr_{X_{1}}[s]-Pr_{X_{2}}[s]|.$ Two distributions $X_{1},X_{2}$
are said to be $\epsilon$-close if $SD(X_{1},X_{2}) \leq \epsilon$.
We now define randomness extractors.
\begin{definition}
  A function $Ext:\{0,1\}^{t}\times \{0,1\}^{r}\rightarrow\{0,1\}^{l}$ is a
  $(t,s,l,\epsilon)$-extractor if for every random variable $X$ over
  $\{0,1\}^t$
  having  min-entropy at least $s$, it holds that $Ext(X,U_{r})$ is $\epsilon$-close
  to $U_l$, where $U_{r},U_l$ denote uniform distributions over $\{0,1\}^{r}$ and $\{0,1\}^l$ respectively.
Further,  we say $Ext$ is an $(t,s,l,\epsilon)$-strong extractor
if:$SD((Ext(X,U_{r}),U_{r}),(U_{l},U_{r}))\leq \epsilon.$ Finally,
$Ext$ is an average-case $(t,s,l,\epsilon)$-strong extractor if the
above
  holds when $\widetilde{\textbf{H}}_{\infty}(X|Y)\geq s$ for some random
  variable $Y$ that is known to the adversary.
\end{definition}

The following lemma shows that there are efficient constructions of
(average case) strong extractor and thus we may utilize a public
random seed to obtain a close to uniform key.

\begin{lemma} [{\cite{DRS04}}]%(Generalized Leftover Hash Lemma,cf.~\cite{DKRS06})
\label{lemma:leftover} Universal hash functions \cite{CW79} are
average case $(t,s,l,\epsilon)$ strong extractors whenever $s\geq
  l+2\log(\frac{1}{\epsilon})-2$.
\end{lemma}

%%%%%%%%%%%%%%%%%%%%%%%%%%%%%%%%%%5

%
%  \begin{eqnarray*}
%\widetilde{\textbf{H}}_{\infty}(A|B,C)
%      & \geq & \widetilde{\textbf{H}}_{\infty}(A,B|C) - \lambda \\
%      & \geq & \widetilde{\textbf{H}}_{\infty}(A|C)-\lambda,
%  \end{eqnarray*}
%\begin{proof}
%proof
%\end{proof}

\ignore{%%%%
In  Figure
\ref{fig:ssketch} we show how the two parties can perform
reconciliation using a secure sketch.
%
%\scalebox{0.65}{%
\begin{figure*}[ht]
\begin{center}
\begin{tabular}{p{0.25\textwidth} p{0.25\textwidth} l}
\texttt{Alice}  && \texttt{Bob} \\
$\rho_A\in \{0,1\}^{n}$ && $\rho_B\in \{0,1\}^{n}$\\

\hline

\\

$s =\texttt{SS}(\rho_A)$ & $\xrightarrow{\hspace{1cm}s\hspace{1cm}}$
&

\\

Return $\rho_A$ & & Return $\texttt{REC}(\rho_B,s)$ \\

\hline
\end{tabular}
\end{center}
\caption{\label{fig:ssketch} Information Reconciliation based on a Secure Sketch. }
\end{figure*}
}%%%%%%%%%%%%%%%%%%%%%%%%%%%%%%

%\begin{lemma}(Generalized Leftover Hash Lemma,cf.~\cite{DKLS06})
%\label{lemma:leftover}
%  For any $\delta>0$, if $Ext$ is a
%  $(n,k-\log\frac{1}{\delta},m,\epsilon)$-strong extractor, then $Ext$ will
%  also be an average-case $(n,k,m,\epsilon+\delta)$-strong extractor.
%  Specifically universal hash functions \cite{CW79} are average case $(n,k,$
%  $m,\epsilon)$ strong extractors whenever $m\leq
%  k-2\log(\frac{1}{\epsilon})+2$.
%\end{lemma}

%%!TEX root = wir03.tex

\section{A Conditional Min-Entropy Learner for Markovian Processes}
\label{sec:learning}

%\noindent{\em Learning conditional min-entropy:} 
Now we proceed to
introduce the central problem for  security argument of any
key-extraction system from measuring physical quantities. Fix a pair
of jointly distributed random variables $\mathbf{X},\mathbf{Y}$, the
problem of learning conditional min-entropy is to
estimate $\widetilde{\textbf{H}}_{\infty}(X|Y)$ from
polynomially many measurements of $\mf{X,Y}$. Unfortunately, in the
general case, the problem is infeasible to solve.
%It is actually not
%hard to see that there exist two random variable pairs,
%$(\mathbf{X,Y})$ and $(\mathbf{X',Y})$, such that from polynomially
%many samples, no algorithm can  distinguish them even though the
%conditional entropy $\widetilde{\textbf{H}}_{\infty}(X|Y)$ and
%$\widetilde{\textbf{H}}_{\infty}(X'|Y)$ have a big difference. For
%instance, $\mathbf{X}=f(\mathbf{Y})$ for a pseudorandom function
%$f$, while $\mathbf{X'}$ is uniform distribution which is
%independent of $\mathbf{Y}$. Note that in the first case, there is
%no uncertainly in $\mathbf{X}$ given $\mathbf{Y}$, while in the
%latter, all the entropy of $\mathbf{X'}$ remains since the condition
%$\mathbf{Y}$ is not related.
As shown in~\cite{valiant11} that at least $O(\frac{N}{\log N})$
many samples are needed for estimating the entropy where $N$ is the
support size. It is easy to see that conditional entropy is even
more difficult to learn, since one can reduce the problem of
learning the entropy to learning the conditional entropy over a
uniformly distributed variable. It follows that in the general case
the conditional min-entropy is not learnable efficiently. It is
therefore important to identify classes of distributions for which
the conditional min-entropy can be learned. Then, when a certain
type of physical measurement can be justified experimentally to
follow the distribution, our security argument will provide strong
evidence regarding the security of key extraction from such
measurements.
%Observe that the conditional
%min-entropy of the key on all the views of the adversary can be
%reduced to the conditional min-entropy of the physical
%measurements of one player , and a reduction
%of entropy loss due to other procedures including message passing
%signal mapping, and etc.
We  identify a wide class of distributions
of measurements below.

Recall that Alice's measurements are denoted by
$\mf{X}=X_1,\ldots,X_n$, while adversary Eve's observations are
$\mf{Y}=Y_1,\ldots,Y_n$.
A simple observation is that the conditional min-entropy can be
calculated directly (via textbook formula) in time polynomial in the
support of $\mathbf{X}$ and $\mathbf{Y}$. This obviously does not
help since the support sets are growing exponentially in $n$.
Note that both measurements $\mathbf{X},\mf{Y}$ are composed of a
sequence of correlated variables with a relatively small support.
A way to simplify the problem is to assume independence
among these variables - however this can be unrealistic.
Here we propose  a more realistic
assumption that the variables are correlated, but it would still allow us to learn the conditional
entropy efficiently. Our main idea is to consider Alice's measurements
$X_1,\ldots,X_n$ as hidden states of a Markov process, while taking
 Eve's observations $Y_1,\ldots, Y_n$ as the visible
output. The following are the basic elements of an HMM (see
\cite{Rabiner89} for a survey):
\ignore{%%%%%%%%%%%%%%%%%%%%%%%%%%%%%%%%
A hidden Markov model is a statistical Markov model in which the
system being modeled is assumed to be a Markov process with
unobserved (hidden) states. In a regular Markov model, the state is
directly visible to the observer, and therefore the state transition
probabilities are the only parameters. In a hidden Markov model, the
state is not directly visible, but output, dependent on the state,
is visible. Each state has a probability distribution over the
possible output tokens. Therefore the sequence of tokens generated
by an HMM gives some information about the sequence of states. Note
that the adjective 'hidden' refers to the state sequence through
which the model passes, not to the parameters of the model; even if
the model parameters are known exactly, the model is still 'hidden'.
}%%%%%%%%%%%%%%%%%%%%%%%%%%%%%%%
\begin{itemize}
\item $k$, number of hidden states
\item $\textbf{S}$, set of states, $\textbf{S}=\{s_1,\ldots,s_k\}$
\item $m$, number of observed symbols
\item $\textbf{O}$, set of observed symbols, $\textbf{O}=\{o_1,\ldots,o_m\}$
\item $\textbf{X}=X_1,\ldots,X_n$, sequence of hidden states
\item $\textbf{Y}=Y_1,\ldots, Y_n$, sequence of observed symbols
\item $A$, the state transition probability matrix, $a_{ij}=\Pr[X_{t+1}=s_{j}|X_{t}=s_{i}]$
\item $B$, the observation probability distribution, $b_{j}(o_{q})=\Pr[Y_{t}=o_{q}|X_{t}=s_{j}]$
\item $\pi$, the initial state distribution, $\pi_{i}=\Pr[X_{1}=s_i]$
\end{itemize}

%\noindent{\bf Assumptions:} Based on the above model, and
%observations from the experiments, the following assumptions can be
%derived about measurements $\{X_i, Y_i\}$, all of which will be
%experimentally falsified.

Given the above we now turn to  the conditional min-entropy calculation. From Definition~\ref{def:ave-entropy}, we get:
\begin{eqnarray}
\widetilde{\textbf{H}}_{\infty}(\textbf{X}|\textbf{Y})
   & = &-\log(E_{y\leftarrow
\textbf{Y}}[2^{-\textbf{H}_{\infty}(\textbf{X}|\textbf{Y}=y)}])\notag\\
   & = &-\log(E_{y\leftarrow \textbf{Y}}[max_{{x}}
\Pr(\textbf{X}={x}|\textbf{Y}=y)])\notag\\
   & = &-\log(\sum_{y\leftarrow \textbf{Y}}[max_{x}
\Pr(\textbf{X}=x,\textbf{Y}=y)])\label{formula:ave-entropy}
\end{eqnarray}

We first calculate the maximum joint probability for a given
sequence of observed symbols $y_1,\ldots,y_n$:
$\max_{x_1,\ldots,x_n}\Pr[\textbf{X}=x_1,\ldots,x_n,\textbf{Y}=y_1,\ldots,y_n].$

%under an extra independent observation assumption, which essentially
%means observed symbols are independent given the corresponding
%states.That is, for any $n$,
%\begin{eqnarray*}
%\Pr[Y_1=y_1,\ldots,Y_n=y_{n}|X_1=x_{1},\ldots,X_n=x_{n}] \\
%   = \prod_{i=1}^{n}Pr[Y_i=y_1|X_i=x_n].
%\end{eqnarray*}

This joint probability can be calculated using the Viterbi algorithm
which is a dynamic programming algorithm we briefly describe below.

Consider the objective function
$\delta_{t}(s_i)=\max\limits_{x_1,\ldots,x_{t-1}}\Pr(X_{1}=x_{1},\ldots,X_{t}=s_{i},Y_{1}=y_{1},\ldots,Y_{t}=y_{t}),$
our goal is then to calculate $\max_{s}\delta_{n}(s)$, with the
following recursive relation:
$\delta_{t+1}(s_j)=\max_{s_i}[\delta_{t}(s_{i})a_{ij}]b_{j}(y_{t+1}),$
which is satisfied by any HMM. The Viterbi algorithm consists of
three key steps:
\begin{enumerate}
\item{Initialization:} $\delta_{i}(s_i)=\pi_{i}\cdot b_{1}(y_1)$
\item{Recursion:} $\delta_{t+1}(s_j)=\max_{s_i}[\delta_{t}(s_{i})a_{ij}]b_{j}(y_{t+1})$
\item{Termination:} $P^{*}=\max_{s}\delta_{n}(s)$
\end{enumerate}

\ignore{%%%%%%%%%%%%%%%%%%%%%%%%%%%%%%5
%%proof of viterbi alg
this can be seen as follows:

\begin{align}
&\delta_{t+1}(s_j)\notag\\
=&\mathop{max}_{x_1,\ldots,x_{t}}\Pr(X_{1}=x_{1},\ldots,X_{t+1}=s_{j},Y_{1}=y_{1},\ldots,Y_{t+1}=y_{t+1})\\
=&{\mathop{max}_{x_{t}}\delta_{t}(x_t)\Pr[X_{t+1}=j,Y_{t+1}=y_{t+1}|x_1,\ldots,x_t,y_1,\ldots,y_t]}\\
=&\mathop{max}_{x_{t}}\frac{\delta_{t}(x_t)\Pr[X_{t+1}=j,y_1,\ldots,y_{t+1}|x_1,\ldots,x_t]}{\Pr[y_1,\ldots,y_t|x_1,\ldots,x_t]}\\
=&\mathop{max}_{x_{t}}\frac{\delta_{t}(x_t)\Pr[y_1,\ldots,y_{t+1}|x_1,\ldots,x_t,j]\Pr[X_{t+1}=j|x_{1},\ldots,x_t]}{\Pr[y_1,\ldots,y_t|x_1,\ldots,x_t]}\\
=&\mathop{max}_{x_{t}}\delta_{t}(x_t)b_{j}(y_{t+1})a_{x_{t}j}\\
=&max_{s_i}[\delta_{t}(s_{i})a_{ij}]b_{j}(y_{t+1})
\end{align}
}%%%%%%%%%%%%%%%%%%%%%%%%%%%%%%%%%

%The proof of the relation relies on the independent observation
%assumption (assumption~\ref{assumption:independent-observation}) as
%well as the fact that $x_1,\ldots,x_n$ are produced by a Markov
%process.
%

It follows that we can calculate the maximum joint probability for a given
observation sequence. Still, to get an estimation for the average case
min-entropy by applying formula (\ref{formula:ave-entropy}), we need
to sum up all such joint probability for all possible observation
sequences. It is very likely that there are exponentially many such
sequences. For example, suppose for any state $s_i$, there are only
3 possible observations, but they are equally possible. In this
case, there would be $3^n$ possible observation sequences with
length $n$.

To circumvent the above, we also study the property of the
distribution of conditional entropy for fixed observed sequences
(i.e., when conditioned on a specific sequence of Eve's
measurements).

Note that the Viterbi algorithm gives only the maximum joint
probability, in order to estimate the conditional min-entropy given
each observed sequence, we need to further get the probability of
appearance for each observed sequence. The forward
algorithm~\cite{BE67} can be used to calculate this quantity in an
HMM. For a given observed sequence $y_1,\ldots,y_n$, we define the
objective function as:
$\alpha_{t}(s_i)=\Pr[Y_{1}=y_{1},\ldots,Y_{t}=y_{t},X_t=s_i],$\ then
the probability of the sequence
$\Pr[Y_{1}=y_{1},\ldots,Y_{n}=y_{n}]=\sum_{i}\alpha_{n}(s_i).$ The
forward algorithm proceeds as:
\begin{itemize}
\item{Initialization:} $\alpha_{1}(s_i)=\pi_{i}\cdot b_{i}(y_1)$
\item{Induction:} $\alpha_{t+1}(s_j)=[\sum_{i}\alpha_{t}(s_{i})\cdot a_{i,j}]b_{j}(y_{t+1})$
\item{Termination:} $P=\sum_{i}\alpha_{n}(s_i)$
\end{itemize}

If we assume that the conditional min-entropies on frequently appearing
observations are stably distributed,
%(jumping ahead, we will show in
%Figure~\ref{figure:stable-entropy}, and
%Figure~\ref{figure:linear-entropy} that this is
%  the case observed in our experiments).
  then, the average case
min-entropy can be approximated by the average of conditional
min-entropies on fixed observations. In other words,
$E_{y\leftarrow\textbf{Y}}[2^{-\textbf{H}_{\infty}(\textbf{X}|\textbf{Y}=y)}]$
will be almost the same as the value
$2^{-\textbf{H}_{\infty}(\textbf{X}|\textbf{Y}=y)}$ for each $y$,
and thus we can approximate
$\widetilde{\textbf{H}}_{\infty}(\textbf{X}|\textbf{Y})$ as the
average of $\textbf{H}_{\infty}(\textbf{X}|\textbf{Y}=y)$ from all
the measurements $y$ we observe.
%it is reasonable to estimate the
%conditional min-entropy in the average case with only a few
%measurements

Now, suppose we perform a series of $M$ experiments, each one with
$n$ measurements. For the $j$-th experiment, we record the
measurements of Eve, and we use them together with the HMM
parameters to calculate the maximum of conditional probability as
$\frac{P^*_j}{P_j}$, where $P^*_j$ is the maximum joint probability given an
observed sequence, returned by the Viterbi algorithm. The
conditional min-entropy for this observed sequence is
$-\log(\frac{P^*_j}{P_j})$. In this way, we will use the average of these
entropies to approximate the average case min-entropy.

To summarize, with the following channel abstractions, one can learn
the conditional min-entropy from only a reasonable number of
experiments.
\begin{enumerate}
\item
\label{assumption:markov}First, the measurements are produced from a
memoryless Markov process, that is, $\forall i,j$, we have:
$\Pr[X_i=x_{1}|X_{i-1}=x_{0}]=\Pr[X_j=x_{1}|X_{j-1}=x_{0}].$
\item
\label{assumption:stationary-transit}Second, the observation
probability does not change over time: specifically, for all
different times $i,j$, $\Pr[Y_i=y|X_i=x]=\Pr[Y_j=y|X_j=x]$.
\item
\label{assumption:independent-observation}Further, observations are
independent, i.e., $\forall n$,
$\Pr[Y_1=y_1,\ldots,Y_n=y_{n}|X_1=x_{1},\ldots,X_n=x_{n}]=\prod_{i=1}^{n}\Pr[Y_i=y_i|X_i=x_i].$
\item
\label{assumption:stationary-entropy} Finally, stable conditional entropy
assumption: we assume that the
random variable ${\mf{H}}_{\infty}(\mf{X|Y=y})$ has small standard
deviation (hence its expectation can be calculated via a small
number of experiments). 
\end{enumerate}

The advantage of our approach is that in concrete
application scenarios, the conditions can be experimentally
justified. Then,  if the physical measurements do satisfy these
properties, we can apply our method to approximate the conditional
min-entropy thus laying a basis for secure extraction following the
protocol and analysis we present in the next section.

%\smallskip
%\noindent{\em Remark} that we show one candidate in our conditional
%min-entropy learner algorithm that lies between the theoretical yet
%unrealistic analysis, and efficient but unsound analysis. We make
%two assumptions about the physical measurements, one is the channel
%abstraction of HMM model, and the other is the distribution of
%conditional entropy on specific eavesdropped measurements does not
%have significantly big biases. Under these two assumptions, we can
%calculate the conditional min-entropy efficiently from the samples
%collected in a small number of experiments. At the same time, we
%will experimentally justify these two assumptions. Of course, one
%can improve the results of learning the conditional entropy through
%more complicated but falsifiable modeling, we will leave this as one
%of the most central open problems for future research.

\ignore{%%%%%%%%%%%%%%%%%%%%%%%%%%%%%%%%%%%%%%%%%%%
%%%%%%sum to estimate average case conditional entropy
It follows that if we have the parameters of the HMM that expresses
the measurements of Alice and Eve over the channel, then we can use
the Viterbi algorithm to estimate the conditional min-entropy of
Alice's measurements as follows. We first perform a series of $M$
experiments each one with $N$ measurements each. For each one, say
the $j$-th experiment, we record the measurements of Eve, and we use
them together with the HMM parameters to calculate the conditional
max probability as $P^*_j = \max_{s}\delta_{n}(s)$. Finally we
estimate the conditional min-entropy as $H^* = -\log ( \sum_{j=1}^M
P^*_j)$.
}%%%%%%%%%%%%%%%%%%%%%%%%%%%%%%%%%%%%%%%%

\section{Application: Physical Layer Key Extraction with Experimental Security Analysis}
\label{sec:wireless-protocol}

We now provide a  key extraction protocol from wireless
signals and analyze its security using the methodology proposed
above. Furthermore, we instantiate it with concrete parameters enabling the two parties to derive e.g., a 128-bit key.
%The justifications for all the assumptions we make using real
%experimental data are in the appendix.

\subsection{Key Extraction Protocol from Wireless Signal Envelope}

Two nodes Alice and Bob execute a ping-like protocol, they send each
other messages through a wireless channel and measure the received
wireless signal strength.  Then they apply the general protocol of the previous
section to those measurements to
extract a key which will later be used for their secure
communication. Note that since the actually messages communicated do not
affect the exchange, the measurements  can be collected during a regular
plaintext communication between Alice and Bob. This means that key-exchange
in this setting can piggyback on top of a plaintext communication.
Our concrete key extraction protocol from received wireless signal
strength is presented in Figure~\ref{figure:protocol}.

\begin{figure*}[t]
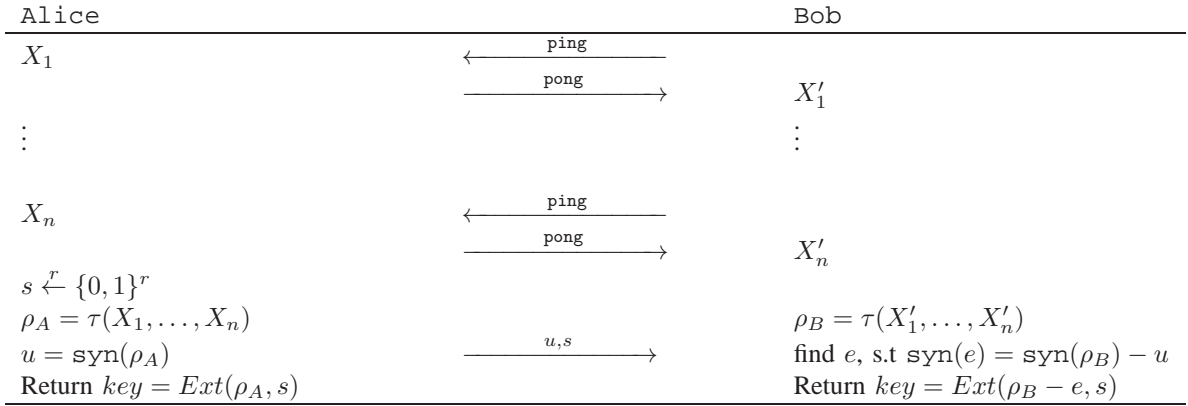

\begin{center}
\begin{tabular}{p{0.3\textwidth} p{0.22\textwidth} l}
\texttt{Alice}  && \texttt{Bob} \\
\hline

$X_1$ & $\xleftarrow{\hspace{1cm}\mathtt{ping}\hspace{1cm}}$ &
\\
 &
$\xrightarrow{\hspace{1cm}\mathtt{pong}\hspace{1cm}}$ &
$X'_1$ \\

\vdots   & &
\vdots \\

$X_n$ & $\xleftarrow{\hspace{1cm}\mathtt{ping}\hspace{1cm}}$ &
\\
 &
$\xrightarrow{\hspace{1cm}\mathtt{pong}\hspace{1cm}}$ &
$X'_n$ \\

$s \stackrel{r}{\leftarrow} \{0,1\}^{r}$ & & \\
$\rho_A = \tau(X_1,\ldots,X_n)$ & & $\rho_B = \tau(X'_1,\ldots,X'_n)$ \\
$u =\texttt{syn}(\rho_A)$ &
$\xrightarrow{\hspace{1cm}u,s\hspace{1cm}}$ &
 find $e$, s.t $\texttt{syn}(e) = \texttt{syn}(\rho_B)-u$
\\

Return $key=Ext(\rho_A,s)$ & & Return $key=Ext(\rho_{B}-e,s) $ \\

\hline
\end{tabular}
\end{center}

\caption{\label{figure:protocol} Our  physical-layer key extraction.
$\texttt{syn}$ is the algorithm computing the syndrome w.r.t. a
public error-correcting code, $\tau(\cdot)$ is a low-distortion
embedding into Hamming metric space, $s$ is the short random seed
for extractor $Ext$. }
\end{figure*}

\subsection{Experimental Setup}
\label{subsec:experiment-setup}

The hardware basis for our experimental platform was the Crossbow MICAz sensor
mote, which contains a Chipcon CC2420 IEEE 802.15.4-compliant RF transceiver.
The RF transceiver in our experiments operated at a frequency of 2.48 GHz and
with a data rate of 250 kbps. Transmit power was set to the maximum level of 0
dBm for each MICAz device in our test network.

The software portion of our experimental platform primarily
consisted of a {\tt ping}-like application that we implemented on
top of the TinyOS operating system. Our implementation consists of
three components: a base station (Alice), a mobile client (Bob) and
an eavesdropper (Eve). Bob sends a {\tt PING} frame
containing a 4-byte sequence number to Alice, who records the
sequence number and the received frame's RSSI value as computed by
the CC2420 transceiver. Alice then sends a {\tt PONG} frame to Bob
containing the sequence number from Bob's original {\tt PING}
request. Bob records the response frame's RSSI value, waits for
1300ms, increments the sequence number and repeats the above
process. If Bob does not receive a response within 500ms of
transmitting his ping request, he will increment the current
sequence number and transmit a new ping request. The 1300ms delay
between {\tt PING} frames sent by Bob is intended to ensure that our
assumption of independence between non-consecutive channel samples
apart holds true, which we  experimentally evaluate  in
Section~\ref{subsec:channel-assumptions}.

%The
%reasoning behind this reduction in precision is explained further
%in Section~\ref{subsec:estimate-uncertainty}.

The eavesdropper implementation, Eve, simply listens passively for any frames
transmitted by any other nearby nodes. Eve records the source and destination
addresses, frame type (i.e., {\tt PING} or {\tt PONG}), sequence number and
observed RSSI value for that frame. After the experiment is manually
terminated, the recorded data is uploaded from each sensor mote to a laptop
for further processing and analysis. We resolve lost frames during
post-processing by computing the intersection between frame sequence numbers
observed by Alice and those observed by Bob. In practice, Alice and Bob could
use an acknowledgement-based protocol (e.g., TCP) to detect lost frames;
however, we opted for a post-processing approach for simplicity of
implementation.

\begin{figure}
    \centering
    \includegraphics[width=3in,height=1.5in]{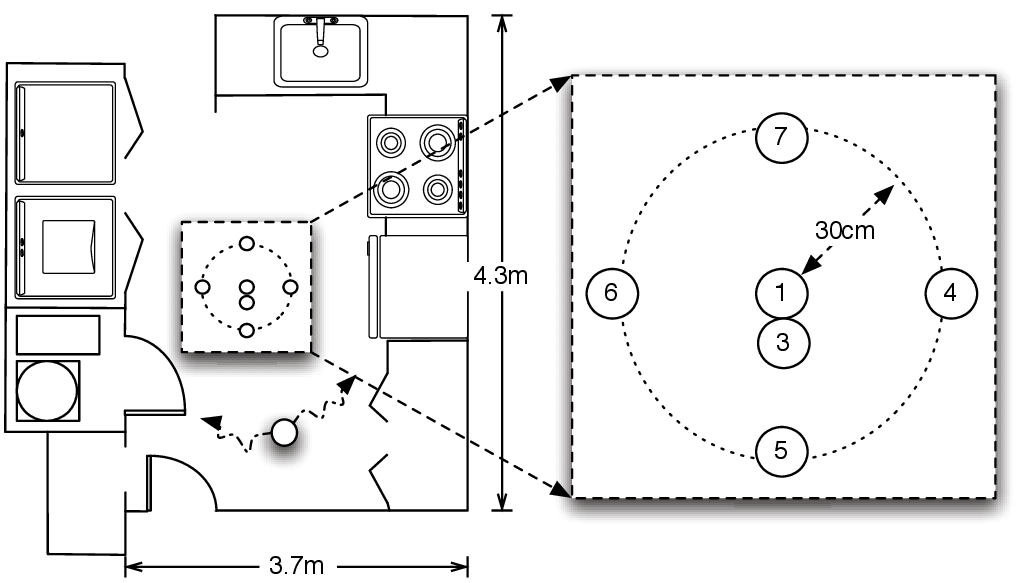}
    %\resizebox{3.25in}{!}{\includegraphics{node-layout}}
    \caption{Floorplan for our indoor test environment and layout of sensor
    motes. Node 1 is the stationary base station (Alice) and Node 2 (Bob) is
    fixed to a robotic platform that moves randomly throughout the test
    environment. The remaining five nodes are passive eavesdropper nodes (Eve)
    that simply record all observed frames and their observed RSSI values.}
    \label{fig:node-layout}
\end{figure}

We conducted two separate data collection experiments in a basic indoor
residential environment. In each experiment, we placed the stationary base
station (Alice) in the middle of the environment with five eavesdropper nodes
placed around the base station as shown in Figure~\ref{fig:node-layout}. The
mobile node (Bob) was mounted to a purpose-built robotic platform that was
programmed to move throughout the environment in a random pattern for the
duration of the experiment at a rate of approximately 14 cm/sec while avoiding
obstacles. The robot would move forward for a period of time selected
uniformly at random from between 5 and 15 seconds. At the end of the randomly
selected travel time, the robot would turn at an angle selected uniformly at
random from between -75 to 75 degrees indicating a left or right turn,
respectively. If the robot detected an obstacle via its front-mounted bump
sensors or ultrasonic sensor before the time expired, the robot would
similarly choose a new direction and travel time at random. The travel time
and turn angles were selected primarily based on the small size of the test
environment.

Each sensor mote can store up to 25,000 channel measurements and so, with a
1300ms delay between {\tt PING} frames, each experiment represents
approximately nine hours of continuous movement over a total distance traveled
of around 4 kilometers per experiment. The data for each node from both
experiments were concatenated together resulting in slightly less than 50,000
samples for each node after accounting for lost or dropped frames. Note that
the number of samples we collected is for the purpose of our security
analysis; an actual key can be computed much faster. See
Section~\ref{sec:instantiation} for more details on the number of samples required for key
extraction.

%%!TEX root = wir03.tex

\subsection{Security Analysis of Our Key Extraction Protocol}
\label{subsec:initial-entropy}

Each experimental dataset was downsampled by six samples in order to
simulate a greater idle time between {\tt PING} frames and ensure
that a channel sample $X_i$ is correlated only with
$X_{i-1},X_{i+1}$, the reasoning will be explained in
Section~\ref{subsec:channel-assumptions}. After downsampling, we
still have around 8100 samples in the experiment. We took 8000
samples and divided them into 100-sample slices, each slice
representing a single {\em key experiment}.

%From the experimental data, all channel measurements are negative
%integers, within the range $[-46,-1]$. So, we use 64-bits unary
%encoding that represents every measurement $x$ as $64-|x|$ bits $0$s
%together with $|x|$ bits $1$s. Then every 100-sample slice becomes a
%6400 bits string. The average $L_1$ distance between Alice and Bob's
%sample slice is only around 58, which is exactly the Hamming
%distance, this will result in a error rate lower than 1\%.
%\medskip
\noindent{\bf Learning the conditional min-entropy:} First, we
estimate the average case min-entropy of Alice's channel
measurements given a single adversary's view using the methodology
proposed in Section~\ref{sec:learning} by incorporating the data
from the key experiments. In our setup, there are 5 adversarial
nodes and we calculated the entropy for each one
separately. All our measurements were calculated for 32 distinct
signal levels.
%Further, we dropped the least 3 significant bits of each channel
%sample for all involved parties, since higher precision experiments
%would have required an order of magnitude more time in data
%collection for the same statistical significance.
%
Moreover, to make our adversarial model stronger, we simply ignored
for both Alice and Bob all samples that were not also observed by
Eve. By doing so, we ensure that we could obtain security
without relying on frames missed by the eavesdropper.

For every experiment of 100 samples, we first compute the maximum
conditional probability of the observed sequence by using the
Viterbi algorithm and forward algorithm as explained in
Section~\ref{sec:learning}; After getting the conditional min-entropy for
the 80 key experiments, we average them to approximate the
average case min-entropy due to the assumption that the
conditional min-entropies on specific observations are stably
distributed (experimental justification will appear in
Figure~\ref{figure:stable-entropy}).

%We prepare all the parameters for the HMM as follows:

Since $X_i$, the random variable representing Alice's $i$-th sample,
follows the same distribution for all $i$ (will be justified in
Section~\ref{subsec:channel-assumptions}), we get the initial state
distribution $\pi$ by counting throughout all channel samples. We
prepare the state transition probability matrix $A$ by estimating
entry $a_{ij}$ by $\frac{n_{j}}{n_{i}}$, where $n_j$ is the number
of measurements equal to $s_i$ such that the next measurement equals
$s_j$, and $n_i$ is the number of measurements equal to $s_i$.
Similarly, we approximate entry $b_{j}(k)$ for the observation
probability distribution by $\frac{m_k}{m_j}$ where $m_j$ is the
number of Alice's measurements equal to $s_j$, and $m_k$ is the
number of Eve's observations equal to $o_k$ while the corresponding
measurements of Alice is $s_j$. This is reasonable since the
observation probability does not change (the independent observation
assumption).

With all these parameters, we next apply the Viterbi algorithm and
forward algorithm to compute the maximum conditional probability
$\mathop{\max}_{x_1,\ldots,x_{100}}\Pr[X_1=x_1,\ldots,X_{100}=x_{100}|
y_{1},\ldots,y_{100}],$ for every experiment, where
$y_{1},\ldots,y_{100}$ are Eve's measurements for a 100-sample
slice. We then apply minus logarithm to those probabilities, and
average them to get our final approximation for the average case
min-entropy contained in 100 samples of Alice. Comparisons for the
five nodes are shown in the sixth column of
Table~\ref{table:estimation}.
%the
%conditional min-entropy on each observed sequence is stably
%distributed as shown in figure~\ref{figure:stable-entropy},

\begin{table}
\centering
\begin{tabular}{|c|c|c|c|c|c|r|r|r|}
  \hline
  Node & $M$ & $n$ & $N$ &  $\bar{d}$ & H & $R_e$ & $R$ \\ \hline
  3 & 80 &100& 800 & 4.42  & 88.40 &0.55\%& 11.05\%\\ \hline
  4 & 80 &100& 800 & 4.29  & 99.82  &0.54\%&12.48\% \\ \hline
  5 & 80 &100& 800 & 4.36  & 98.89   &0.55\%&12.36\%\\ \hline
  6 & 80 &100& 800 & 4.28  & 93.07   &0.54\%&11.63\%\\ \hline
  7 & 80 &100& 800 & 4.35  & 94.08  &0.54\%&11.76\%\\ \hline
\end{tabular}
\caption{\label{table:estimation} Comparison of conditional
min-entropy for the 5 eavesdropping  nodes.  $M$: \#
of experiments, $n$: \# of channel samples in each experiment,
$N=|\rho_A|$: the length of bit string after bit quantization,
$\bar{d}$: the average \# of word errors, $H$: the average
case min-entropy, $R_e=d/|\rho_A|$: average error rate, and $R$: average entropy rate. The node IDs are
as shown in Figure~\ref{fig:node-layout}.}
\vspace{-8mm}
\end{table}

\noindent{\bf Accounting for Entropy Loss:} With a reasonable amount
of entropy contained in Alice's channel measurements, conditioned on
adversary's view, we now proceed to calculate the entropy loss for
each step, and subtract those losses to get the final result about
security of the key.

First, for the bit quantization step, the way we defined   the
embedding $\tau$, it is a one-to-one map so, obviously, there is no
entropy loss during this step. Thus,
$\widetilde{\textbf{H}}_{\infty}(\rho_A|\textbf{Y})=\widetilde{\textbf{H}}_{\infty}(\textbf{X}|\textbf{Y}).$

Next we consider the entropy loss in the information reconciliation
step. Suppose an $(n,k)$ binary code is used to construct the
syndrome based secure sketch $u$ described in the protocol shown in
Fig~\ref{figure:protocol}, a total of $|u|=n-k$ bits are
transmitted during the information reconciliation step. Thus,
$$\widetilde{\textbf{H}}_{\infty}(\rho_A|\textbf{Y},u)=\widetilde{\textbf{H}}_{\infty}(\mf{X}|\textbf{Y},u)\geq\widetilde{\textbf{H}}_{\infty}(\textbf{X}|\textbf{Y})-|u|.$$

Finally,
% if there is enough entropy left, and we want to get a key
%with length $t$ and achieves  a security level of $\epsilon_u$,
%i.e., the distribution of the $key$ is $\epsilon_u$ close to the
%uniform distribution over $\{0,1\}^l$,  
according to
lemma~\ref{lemma:leftover}, if
$\widetilde{\textbf{H}}_{\infty}(\rho_A|\textbf{Y},u)\geq s\geq
l+2\log\frac{1}{\epsilon_u}$, we can apply a $(t,s,l,\epsilon_u)$
strong average case extractor 
%such as a universal hash 
to extract a key of length $l$ with distribution $\epsilon_u$ close to the uniform distribution over $\{0,1\}^l$.

Summarizing above, we derive a sufficient condition for extracting a
secure key. Note that as long as
the error rate is not too high, and there is a non-zero entropy left
after the phase of information reconciliation, it is always possible
to amplify the entropy by collecting more samples before the bit
quantization step, and leaves room for privacy amplification. 

%%%%%%%%%%

%\input{analysis-assumptions.tex}

%%!TEX root = wir03.tex

\subsection{Instantiation of Our Protocol}
\label{sec:instantiation}

In this section, we will  instantiate our general protocol
for  physical layer key extraction with concrete parameters from the experimental data.

We first study how conditional entropy grows in a
block of samples with different size. By taking 10 samples as a
unit,  we slice the 8000 samples into 40 pieces, each representing
an individual experiment with 200 samples. We then calculate
conditional min-entropy (of 10 samples, 20 samples, and so forth
until 200 samples) using our methodology by averaging among these 40
experiments. In this way we get the plot of the average conditional
min-entropy with standard deviation of the 40 values as shown in
Figure~\ref{figure:linear-entropy}. We use the method of least
squares to get an approximation of the average case min-entropy as a
linear function of the number of samples and obtain the following
equation: $g(x)=\frac{985}{1000}x+\frac{1467}{1000},$ i.e.,
there is $g(x)$ bits average case min-entropy in every $x$ samples.

\begin{figure}
    \centering
    \includegraphics[width=3in,height=1.5in]{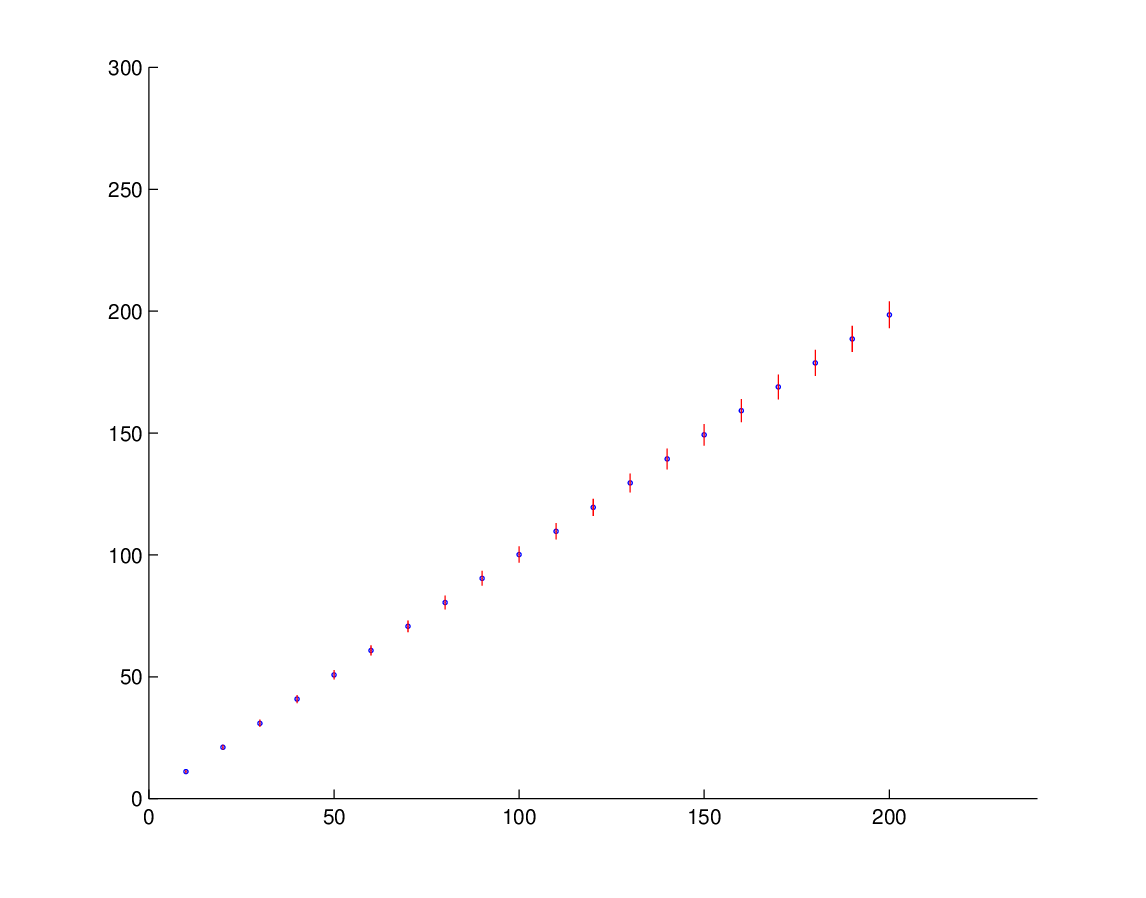}
    %\resizebox{3.25in}{!}{\includegraphics{node-layout}}
    \caption{Linear growth of conditional min-entropy for node4, x-axis is the
    number of samples in the sequence, y-axis is the value of conditional min-entropy,
    the (blue) points are average of conditional min-entropy in 40 experiments, the (red) lines start from
    mean-standard deviation and extend to mean+standard deviation.}
    \label{figure:linear-entropy}
\end{figure}

Similarly, we study the distribution of word (8-bit block) errors.
We first divide the bitstrings $\rho_A, \rho_B$ in  8-bit blocks,
and count the number of different blocks between $\rho_A,\rho_B$.
The result is presented in Figure~\ref{figure:linear-error}. In this
way we obtain an approximation of the average word error $e$:
$e(x)=\frac{43}{1000}x+\frac{48}{1000},$ which means there are
$e(x)$ word errors on average in $x$-sample experiment.
\begin{figure}
    \centering
    \includegraphics[width=3in,height=1.3in]{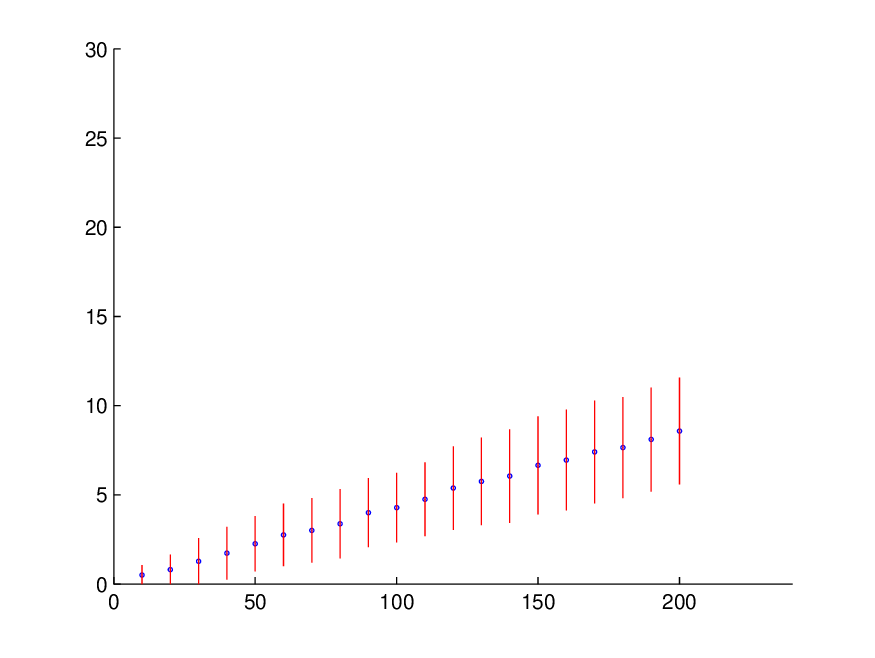}
    %\resizebox{3.25in}{!}{\includegraphics{node-layout}}
    \caption{Distribution of errors, x-axis is the number of
    samples, y-axis is the number of errors, the (blue) points are average
    among 40 experiments, the (red) lines start from mean-standard deviation to
    mean+standard deviation}
    \label{figure:linear-error}
    \vspace{-8mm}
\end{figure}

Our results experimentally show the linear growth of conditional
min-entropy and that of errors. We can now provide an explicit
instantiation of our protocol.

\begin{theorem}
Suppose $n$ is number of samples in our protocol. Using a
Reed-Solomon code over $F_{2^8}$ for the sketch, when $n\geq
\max\{12.54\lambda+6.27l-3.24, 2326c-1\}$, we can successfully
implement a $(l,e^{-c},2^{-\lambda})$ key generation system.
\end{theorem}
\begin{IEEEproof}
Given the estimated conditional entropy rate as
$g(n)=\frac{985}{1000}n+\frac{1467}{1000}$, and error rate as
$e(n)=\frac{43}{1000}n+\frac{48}{1000}$, we choose a RS code over
$F_{2^{8}}$ which can correct up to $\frac{6}{5}e(n)$ errors (slice
the bitstring into proper sized blocks if needed). Thus, the entropy
loss of applying the code would be $8\cdot2\cdot\frac{6}{5}e$, the
entropy loss of applying extractor is $2\lambda$, and similar to the
analysis at the end of section~\ref{subsec:initial-entropy}, we
immediately get: if $159.4n+549.4>1000l+2000\lambda$, we will have a
$l$-bits key with $2^{-\lambda}$ distance to $U_l$. For correctness,
from the Chernoff bound, the probability of having $\frac{6}{5}e(n)$
errors is $e^{\frac{-e(n)}{100}}$, we can easily derive another
condition.
\end{IEEEproof}

To put the above into perspective,  if we set
$\epsilon_{c}=e^{-1},\epsilon_{u}=2^{-80}$,  we can do the
following: first we make around 1800 measurements, then use a 8-bit
unary encoding on the absolute value. (We take only 5-bit precision for the RSSI value.)
%(Since we take 5-bit precision, and in our
%experiments almost all (99.95\%) RSSI fall in [-8,0]. Also, we
%ignore the bit for sign, as they are all negative). 
Finally, we
employ a (255,229)-RS code for secure sketch and apply 2-universal
hash functions~\cite{CW79} as randomness extractor and harvest a
128-bit key. Note that the success rate of the exchange is at least
$1-\frac{1}{e}$ and is detectable by the parties hence in case
of failure the extraction is repeated.

%%%%%%%%%%%%%%%%%%%%%%%%%%%%%%%%%

\ignore{%%%%%%%%%%%%%%%%%%%
%%%%%%%%%%%%% time of completing a key agreement
Note that in our experimental setup this would be achieved with 256
samples and utilizing the 7800ms delay it would require about 33
minutes to extract the key (recall that the purpose of the delay is
to de-correlate successful measurements so that our security
analysis applies).
}%%%%%%%%%%%%%%%%%%%%%5

\subsection{Justification of Channel Assumptions}
\label{subsec:channel-assumptions}

We now show that  all channel assumptions stated in
Section~\ref{sec:learning} hold  based on the
analysis of the experimental data collected in the test environment
described in Section~\ref{subsec:experiment-setup}.

%\medskip
\noindent{\bf Stationary Memoryless Markov Process
Assumption:} We split assumption~\ref{assumption:markov} to three
sub-assumptions. Let us begin with the assumption that every channel
measurement $X_i$ depends only on the previous measurement $X_{i-1}$
for all $i$. As mentioned in our description of our test platform,
we used an artificial delay of 1300ms between ping frames in order
to induce independence between channel samples greater than one
sample apart. To evaluate this assumption of independence we
performed the following analysis. We selected some $i \in [1+m,n]$
uniformly at random, where $n$ is the number of samples in a node's
measurements, and $m = 500$. We then created a $(m+1)$-length row
vector as $\{X_{i}, X_{i-1}, \ldots, X_{i-m}\}$, where $X_{i-m}$ is
the $(i-m)$-th sample from a node's channel measurements. This was
repeated 10,000 times, concatenating each of the 10,000 row vectors
together to create a $10000 \times (m+1)$ matrix $A$.

The Pearson's linear correlation coefficients of the resulting
column vectors in $A$ were then computed to find the correlation of
$X_i$ with $X_{i-k}$ for each $k \in [1,500]$. From the correlation
analysis of our experimental data, we found that there exists a
statistically  correlation between samples $X_i$ and
$X_{i-k}$ at a significance level of $\alpha = 0.05$ for values of
$k$ up to $6$. \footnote{This clearly shows that previous idealized assumptions that the signals are independently distributed are not precise.} Since we want to ensure that $X_i$ and $X_{i-k}$ are
independent for $k > 1$ for our channel assumptions to hold, we
further downsample the collected data by a factor of six in order to
induce this level of independence (i.e., ensuring at least 7800ms
between samples). The remainder of our experimental analysis is
performed on this downsampled data. 

%The correlation among signals
%under downsampled data is shown in Figure~\ref{fig:downsample}.

%\begin{figure}
%    \centering
%    \includegraphics[width=2.8in, height=1.3in]{downsampled-correlation.eps}
%    %\resizebox{3.25in}{!}{\includegraphics{node-layout}}
%    \caption{Correlation among signals in the downsampled data}
%    \label{fig:downsample}
%\end{figure}

%2. Take every $X_i$ as different random variable, they follow
%the identical distribution
\ignore{%%%%%%%%%%%%%%%%%%%%%%%%%
!!!!!!!!!!!!!!\emph{keep justification for $X_i$ follow identical
distribution} }%%%%%%%%%%%%%%%%%%%%%%

%%%%%%%%%%%%%%%%%%%%%%

\begin{figure}
    \centering
    \includegraphics[width=3in, height=1.5in]{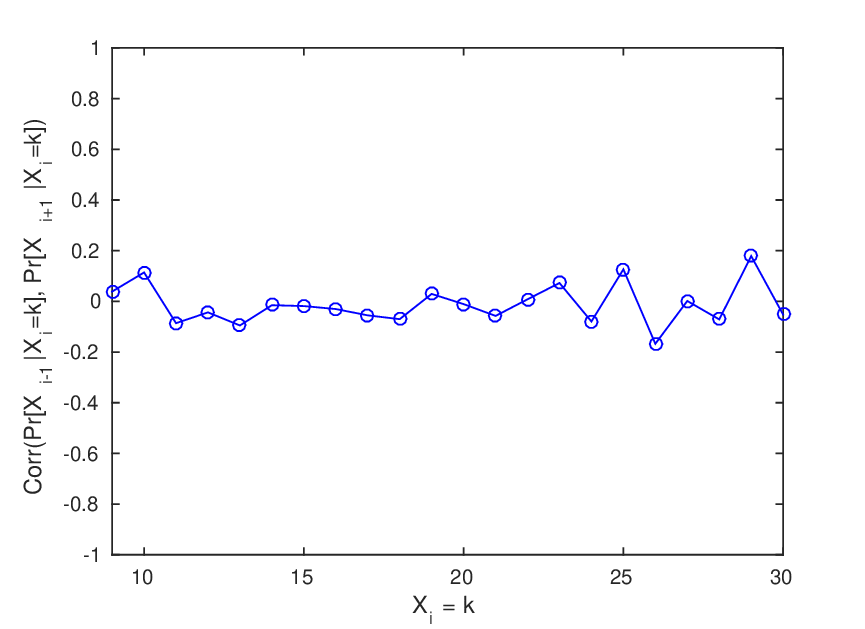}
    \caption{Correlation between $\Pr[X_{i-1}|X_i=k]$ and $\Pr[X_{i+1}|X_i=k]$
             for channel sample values $k \in [9,29]$.}
    \label{fig:correlation}
    \vspace{-6mm}
\end{figure}

We further show that for all channel samples $X_i = k$, where $k$ is a 5-bit
RSSI value, that $\Pr[X_{i-1}|X_i=k]$ and $\Pr[X_{i+1}|X_i=k]$ are uncorrelated.
For every value of $k \in [0,32]$, we derive the distributions
$\Pr[X_{i-1}|X_i=k]$ and $\Pr[X_{i+1}|X_i=k]$ and again compute the Pearson's
linear correlation coefficient between the two vectors. In our experimental
dataset, value of $k < 9$ and $k > 29$ were infrequent and so we were unable to
compute statistically significant correlation coefficients for those values of
$k$. The remaining values of $k \in [9,29]$, however, constitute approximately
95\% of our channel samples. The correlation coefficients between the two 
distributions for each value of $k \in [9,29]$ are shown in
Figure~\ref{fig:correlation}. Our results show that $\Pr[X_{i-1}|X_i=k]$ and 
$\Pr[X_{i+1}|X_i=k]$ are uncorrelated for values of $k$ that have a sufficient
number of samples (approximately 50 or more).

%%%%%%%%%%%K-S test%%%%%%%
\ignore{%%%%%%K-S test
Finally, we consider the assumption that the Markov process
transition probability is stationary. That is, we want to show that
the equality $\Pr[X_i = x_1|X_{i-1} = x_0] = \Pr[X_j = x_1 |
X_{j-1} = x_0]$ holds for any $i$ and $j$ with $i \neq j$. We
randomly selected some $i$ and $j$ in $[1,N-1]$ uniformly at random
and append the channel sample $X_{i+1}$ to a vector $\vec{x}_1$ and
$X_{j+1}$ to a vector $\vec{x}_2$. This was repeated 10,000 times to
create two 10,000-element vectors. We then applied the two-sample
K-S test to $\vec{x}_1$ and $\vec{x}_2$ to determine whether the
sample distributions are identical. We found that in every test case
the resulting distributions are identical at a
confidence level of $\alpha = 0.05$.
}%%%%%K-S test
%%%%%%%%%%%%%%%%%%%%%

\begin{figure}
    \centering
    \includegraphics[width=3in, height=1.5in]{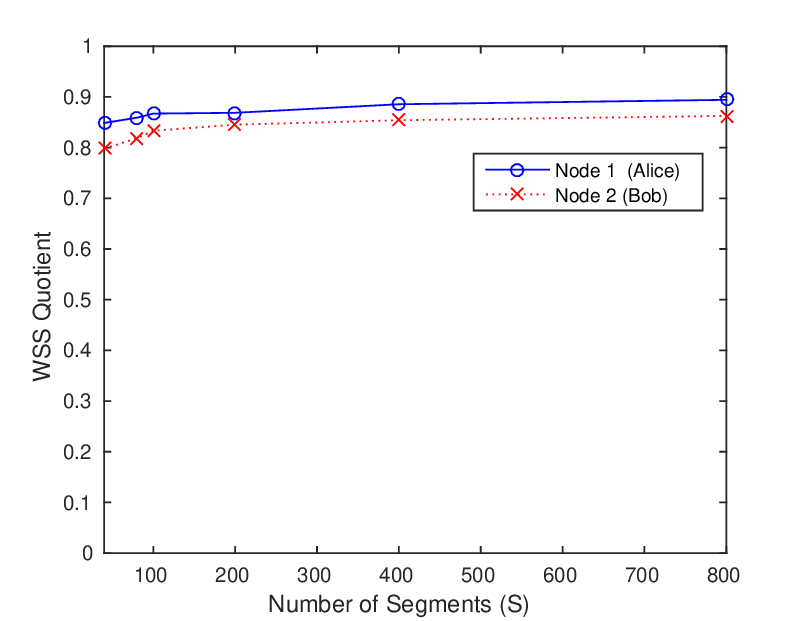}
    \caption{Computed WSS quotient values for varying segment sizes.}
    \label{fig:wss}
    \vspace{-6mm}
\end{figure}

Finally, we consider the assumption that the Markov process
transition probability is stationary. That is, we want to show that: $\Pr[X_i = x_1|X_{i-1} = x_0] = \Pr[X_j = x_1 |
X_{j-1} = x_0],$ holds for any $i$ and $j$ with $i \neq j$. For this
test, we use the wide-sense stationarity (WSS) quotient~\cite{SSO07}
to evaluate whether the first and second moments of our experimental
dataset. We computed the WSS quotient for multiple segment sizes ($S$)
and provide the results in Figure~\ref{fig:wss}. Our results show a
high level of stationarity for the first and second moments of our
sample data across all tested values of $S$. We also performed  a series of K-S tests with similar results. 

%\medskip
\noindent{\bf Stationary Observation Probability Assumption:} Next,
we consider our assumption~\ref{assumption:stationary-transit} that
the channel sample observation probability is stationary.
Specifically, we want to show that $\Pr[Y_i = y | X_i = x] = \Pr[Y_j
= y | X_j = x]$ holds true for any two sample indices $i$ and $j$
with $i \neq j$, where $x$ and $y$ indicate concrete RSSI values
observed by Alice and Eve, respectively. We can easily find the
correlation between $Y_i$ and $X_{i}$ for $k \in [-500,500]$ for
each Eve node using the same approach based on computing the linear
Pearson's correlation coefficient that we described above for
evaluating the first assumption. The results were averaged across
all five eavesdropper nodes. We found that an eavesdropper
measurement $Y_i$ was statistically independent from some sample
$X_j$ collected by Alice for $i\neq j$ at a significance level of
$\alpha = 0.05$.

We use a technique similar to those above to evaluate whether
$\Pr[Y_i | X_i] = \Pr[Y_j | X_j]$ for any two sample indices $i$ and
$j$ with $i \neq j$). We first partition the set of sample indices
$I = \{1,2,\ldots,N\}$ into two equal-sized random subsets
$I^\prime_1$ and $I^\prime_2$ such that $|I^\prime_1| = |I^\prime_2|
= \lfloor n/2 \rfloor$. Eve's channel measurement set $Y$ was
similarly partitioned into two subsets $Y^\prime_1$ and
$Y^\prime_2$, such that $Y^\prime_1$ contains the measurements in
$Y$ corresponding to the channel sample indices in $I^\prime_1$ and
$Y^\prime_2$ contains the measurements in $Y$ corresponding to the
channel sample indices in $I^\prime_2$. Alice's channel measurements
were partitioned into two subsets $X^\prime_1$ and $X^\prime_2$ in
the same manner. We then compute the joint distributions between
$X^\prime_1$ and $Y^\prime_1$, as well as between $X^\prime_2$ and
$Y^\prime_2$. For each RSSI value $x \in [-24,0]$ (the range of
possible RSSI values in our experimental platform), we use a
two-sample K-S test to determine whether the two distributions
$\Pr[Y^\prime_1|X^\prime_1 = x]$ and $\Pr[Y^\prime_2|X^\prime_2 =
x]$ are identical. This process was repeated a total of 10,000 times
for each of the five eavesdropper sets $Y$ (i.e., a total of 50,000
trials). Out of the 50,000 trials, we accepted the null hypothesis
that the two distributions were identical 100\% of the time with a
significance level of $\alpha = 0.05$ and so we conclude that the
assumption $\Pr[Y_i|X_i] = \Pr[Y_j|X_j]$ holds true when $i \neq j$.

%\medskip
\noindent{\bf Independent Observation Assumption:} This states that
the dependency that potentially exists between the adversary's
measurements comes strictly from the existing dependency of the
corresponding measurements of Alice and not any other source. Note
that  from assumption~\ref{assumption:stationary-transit}, each
$Y_i$ follows a distribution that depends only on the value of $X_i$
and this dependency does not change over time. This suggests that
there exists a probabilistic function $F$ over $X_i$ which can
simulate the observations of Eve, i.e., $Y_{i}=F(X_{i})$. In such
case we can prove the independent observation assumption:
\[
\begin{aligned}
&\Pr[Y_{1}=y_{1},\ldots,Y_{n}=y_{n}|x_{1},\ldots,x_{n}]\\
=&\Pr[F(X_{1})=y_{1},\ldots,F(X_{n})=y_{n}|x_1,\ldots,x_n]\\
=&\Pr[F(x_{1})=y_{1},\ldots,F(x_{n})=y_{n}]\\
=&\prod_{i=1}^{n}\Pr[F(x_{i})=y_i]=\prod_{i=1}^{n}\Pr[F(X_{i})=y_{i}|X_{i}=x_{i}]\\
=&\prod_{i=1}^{n}\Pr[Y_{i}=y_i|X_i=x_{i}]
\end{aligned}
\vspace{-2mm}
\]

%\medskip
\noindent{\bf Stable Conditional Min-Entropy Assumption}: The
intuition for assumption~\ref{assumption:stationary-entropy} is that
even though at some location, some observations may have better
advantage to predict the original measurement, the sequence would be
long enough, and thus when averaging across the whole sequence, the
advantage will diminish. To experimentally justify this assumption,
we calculate the min-entropy of the sequence received by node1
conditioned on node4's observation in each experiment (represented
by each 100-sample block). We observe that it is stably distributed
as shown in Figure~\ref{figure:stable-entropy}.
\begin{figure}
    \centering
    \includegraphics[width=2.8in,height=1.3in]{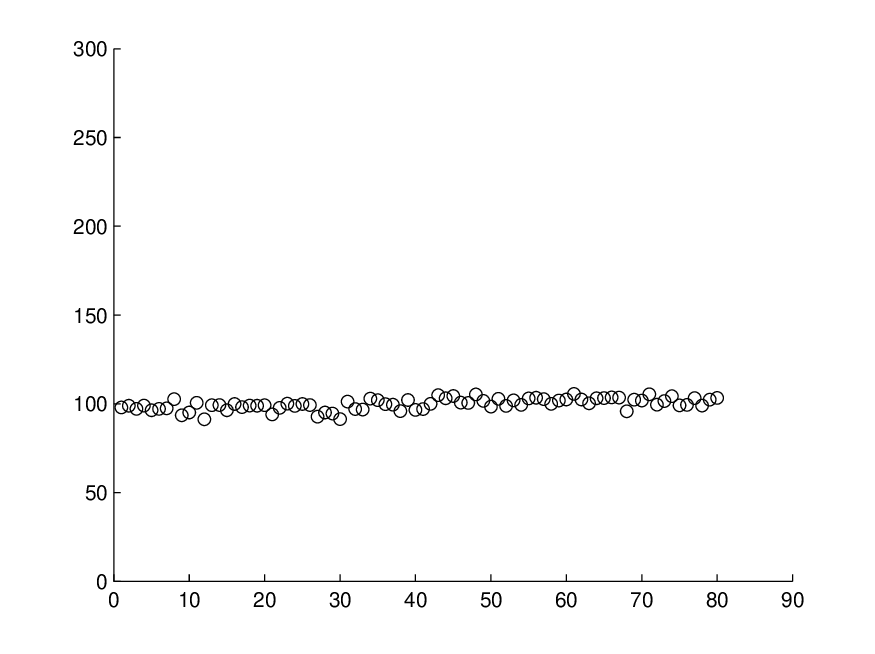}
    %\resizebox{3.25in}{!}{\includegraphics{node-layout}}
    \caption{Distribution of conditional min-entropy in node4, x-axis is the index of
    experiment, y-axis is the value of conditional min-entropy,
    300 is the maximum entropy contained in the 100 samples}
    \label{figure:stable-entropy}
    \vspace{-7mm}
\end{figure}

%Thus  assumption \ref{assumption:independent-observation} is  justified.

%%!TEX root = wir03.tex

\section{Conclusions and Future work}
\label{sec:conclusions}

%In this paper we defined, analyzed and presented our results from the design
% and implementation of an efficient physical-layer key extraction
%protocol that utilizes  correlated
%measurements of the wireless signal envelope between two
%communicating parties.

%Our approach yields several advantages over previous work in that it
%allows for a rigorous experimental security analysis that guarantees
%the agreed key has reasonable level of security. Furthermore, our
%protocol reduces the level of interactivity for bit quantization and
%key reconciliation and does not depend on specialized hardware.

In this paper we presented a framework for a rigorous experimental
security analysis that guarantees a reasonable level of security about the agreed key extracted from
physical quantities, and we show an
application of the framework to key extraction from
wireless signals with concrete parameters based on our
experimental data. Our novel methodology %for the security analysis of physical-layer key extraction
relies on HMM model and the estimation of
conditional min-entropy through a dynamic programming approach that
relies on the Viterbi and forward algorithms.
Our work lays   a comprehensive
methodology for  arguing experimentally the security of
physical-layer key extraction against a
  passive eavesdropper nodes and can be applied to a number of other similar
protocols in a similar way as we demonstrated here.
At a very high level, our methodology entails the following
basic steps, (i) employ a falsifiable channel abstraction,
equipped with an  efficient
conditional entropy approximation algorithm
and the execution of experiments that estimate the conditional min-entropy using the algorithm (ii)
justify experimentally
the channel abstraction, (iii) account for the entropy loss incurred due
to quantization and  information reconciliation.

While the above is a step forward in the
security analysis of physical-layer key extraction there is 
still limited evidence for the security of these protocols in general settings;
 multitude of open questions remain and  our methodology can be
extended in a number of directions: (i) increase the size of the HMM
or even substitute the HMM with more general Bayesian networks to
improve the accuracy of the  estimation of conditional min-entropy
and potentially  decrease the time needed for key generation.
(ii) consider the case of active adversaries and develop protocols
with adversarial interference detection for which
conditional min entropy can still be effectively estimated. (iii)
consider observations from multiple eavesdropper nodes
simultaneously.

%This is work in progress. In our next step we will fully investigate
%the exact values of the leakage advantage $\zeta$ and consider
%various experimental setups to see in what circumstances this probability
%becomes non-trivial. We stress that our methodology can still yield positive
%results in the setting that $\zeta$ is bounded away from $0$.

\bibliographystyle{IEEEtran}

%\bibliography{IEEEabrv,ref}

%\bibliographystyle{abbrv}
\bibliography{ref}

\end{document}